\documentclass[a4paper]{llncs}
\usepackage[utf8]{inputenc}
\usepackage[english]{babel}
\usepackage{amsmath}
\usepackage{amsfonts}
\usepackage{amssymb}
\usepackage{mathrsfs}
\usepackage[vertfit]{breakurl}
\usepackage{hyperref}
\usepackage{graphicx}
\usepackage{cleveref}

\newcommand{\C}{\mathcal{C} }

\newcommand{\A}{\mathcal{A}}
\newcommand{\B}{\mathcal{B}}

\newcommand{\Zp}{\mathbb{Z}_p}
\newcommand{\G}{\mathbb{G}_1}
\newcommand{\Gg}{\mathbb{G}_2}
\newcommand{\Gt}{\mathbb{G}_T}
\newcommand{\egg}{e(g_1,g_2)}

\newcommand{\Fb}{ \mathbb{F}_{2}}

\begin{document}


\title{Public Ledger for Sensitive Data}

\author{Riccardo Longo \and Massimiliano Sala}
\institute{Department of Mathematics, University of Trento, via Sommarive, 14 - 38123 Povo (Trento), Italy\\ \email{riccardolongomath@gmail.com, maxsalacodes@gmail.com}}

\maketitle

\begin{abstract}
Satoshi Nakamoto's Blockchain allows to build publicly verifiable and almost immutable ledgers, but sometimes privacy has to be factored in.

In this work an original protocol is presented that allows sensitive data to be stored on a ledger where its integrity may be publicly verified, but its privacy is preserved and owners can tightly manage the sharing of their information with efficient revocation.
\end{abstract}





\section{Introduction}
In the digital era information sharing is deeply embedded in practically every process: it is difficult to achieve anything without having to share some personal data.
Thankfully modern legislation (e.g. Europe's GDPR) is catching up with the privacy concerns that this kind of workflow fosters, and it is becoming mandatory to allow users to manage how their data is shared.
In particular the \emph{right to be forgotten} allows users to request the deletion of their data from one service's database, so they can revoke the consent to access and handle their informations.

This also means that a platform that manages the sharing of personal data must have an efficient revocation system that allows granular management of users' consent.
Ideally users should be able to grant \emph{one-time access}, that in practice means limited in scope: access is granted as long as it is really needed for the fulfilment of the service.
This property essentially requires automatic revocation.
Should be noted that this approach follows the best practices of \emph{need-to-know} and \emph{principle of least privilege}, especially when paired with \emph{end-to-end} encryption.

Public ledgers based on blockchains \cite{bitcoin} are a great solution for storing public information and they guarantee immutability and accountability.
However serious problems may arise when they are used to store sensitive data, such as health records.
In fact this kind of data is inherently private, but often needs to be shared between multiple service providers (such as healthcare insurance companies, hospitals, pharmacies), possibly with warranties that this data is legit and not been tampered with.

Research on distributed storage has often focused on secure protocols, see e.g. \cite{garay1997secure}, however the advent of GDPR has started a quest for applications with strong privacy properties suitable for sensitive data.
For example there are interesting works that exploit blockchain technologies to enhance access control \cite{azaria2016medrec,zyskind2015decentralizing}.
The work of Kosba et al. \cite{kosba2016hawk} provides a solution to the problem of keeping a smart contract private.
The work of Groth \cite{groth2010short} shows how private data can be encrypted and posted on a blockchain and later prove some properties about it.
\newline

We propose an original protocol that enables the construction of a public ledger that securely and privately stores sensitive information with a blockchain infrastructure.
The intent is to allow secure sharing of this data between authorized parties, integrating an \emph{end-to-end} and \emph{one-time access} system that provides full control over the usage of private data, all while exploiting the accountability and decentralisation of a blockchain, either pre-existing or purposely built.

\subsection{Intuition}\label{intuition}
The work presented here is based on results included in the first author's PhD thesis \cite{longo2018formal}.\\
The approach used to achieve a one-time access system is based on the fact that when a party gains access to a cleartext, it has the opportunity to copy the data and store it locally.
If this is the case, it is worthless to negate further access to that cleartext, since that party already owns a copy.
However it is often expensive to maintain a local copy of the data, so the goal becomes to avoid further access to the data \emph{unless} a local copy has been made.
More precisely the aim is to guarantee that no party has access to information with size smaller than the entire cleartext, that can grant further access to the data (e.g. a key).

So the idea is to encrypt the cleartext with a sort of one-time-pad scheme, in this way the key has at least the same size of the cleartext, so storing it does not give any advantage over storing the cleartext directly.
On the contrary, while plaintext documents are usually reasonably compressible, cryptographic keys are not, given their high entropy.

To achieve practicality the \emph{pad} is stored encrypted on the ledger, so a smaller key is sufficient to recover the original data.
To maintain the one-time property the pad is periodically re-encrypted so a new key is needed to access the data again.
The advantage over a classical \emph{storage re-encryption} is that the protocol combines a shared encrypted pad with individual \emph{per-file keys} to generate the actual encryption pad, so that it is sufficient to re-encrypt the shared pad to effectively revoke access to all data.
Thus the revocation is almost constant-time, without affecting the privacy.

The scenario is the following: a private user $ U $ wants to store its private data safely on the cloud, with the option to share that information with a service provider $ P $ for some time, and with the guarantee that an independent third party could audit the integrity of the data; this cloud is also called \emph{public ledger}.
To coordinate the storage and in particular the revocation mechanism there is a file keeper $ F $ that is not properly a \emph{trusted third party} since it is not responsible for the confidentiality (but is in charge of a correct revocation), moreover the party acting in this role could change at every revocation interval and its behaviour can be publicly audited.

After the initialisation of the public ledger, $ U $ sets up its keys (public and private) according to the parameters of the ledger.
The ledger is divided in a \emph{constant part} where the encrypted data resides, available for sharing and public auditing, and a \emph{variable part} used to coordinate encryption, decryption and revocation; $ F $ is in charge of updating and maintaining the variable portion of the ledger.
To optimise storage costs encrypted data could be actually stored on a cheaper location suitable for remote access and sharing (such as a cloud or distributed file system) and preserve on the public ledger only some metadata used to publicly audit the integrity of the ciphertexts and the File Keeper's behaviour.
Note that conventional blockchain ledgers are immutable (in the sense of\emph{ append-only write access}), so the concept of a \emph{variable ledger} may seem a contradiction in terms.
However we only require this portion of the ledger to be publicly readable by all parties, while it is maintained (i.e. writeable) by $F$ only, so it is more akin to a traditional website and could easily be conceived and deployed as a separate entity from the \emph{constant ledger}.

To store information safely on the ledger, $ U $ needs an \emph{encryption token} given by $ F $, that is used to synchronise with the updating part of the ledger.
The token is combined with the user's private key, a session key, and additional parameters available on the variable portion of the ledger, to encrypt $ U $'s personal data, that is then safely stored on the constant part of the ledger alongside additional information used to publicly audit the integrity of the data and the goodness of $ F $'s behaviour in keeping the variable ledger.
To offload key management from the user the session key is encapsulated (i.e. safely masked) using $ U $'s private key and stored in the variable part of the ledger.

To share data with the service provider $ P $, the user $ U $ retrieves the relevant encapsulated key, unlocks it using its own private key and shares it with $ P $.
This \emph{unlocked} key, combined with additional parameters available on the variable ledger, can be used to successfully decrypt the data as stored on the static part of the ledger.
When the ledger is updated by $ F $ this key is automatically revoked, because the decryption depends on the combination of this key with the parameters of the variable ledger, that have changed with the update.
As aforementioned this means that with a single update every key is revoked, but users can easily re-grant access to any document by unlocking the relevant encapsulated key (that has also been updated by $F$).

Moreover since the data has not been re-encrypted in a conventional way, its auditing is straightforward: it is a classic verification of static (immutable) data, which is impossible to achieve with a conventional re-encryption.
\newline

The construction of the protocol is inspired from one-time-pad encryption, public-key encryption schemes based on bilinear pairings, and the Diffie-Hellman key exchange. 
The novelty of the approach resides in the modelling of the scenario of tightly-managed information sharing alongside its security requirements, and the actual construction of the protocol, that carefully combines the underlying primitives to achieve the desired properties.
In particular to the best of our knowledge this is the first construction with the one-time-access property and fast revocation, combined with full support of public auditing.


\section{Requirements and assumptions}
This section covers background information necessary to understand the protocol presented in this work and its security.

In particular, some mathematical notions about bilinear groups are given, alongside the cryptographic assumptions that will be used.
 \subsection{Bilinear Maps}\label{pbg}
Pairing-based cryptography exploits the properties of bilinear pairings to add new functionalities to encryption schemes that are difficult or impossible to achieve with classical primitives.
Bilinear groups are the main environment for this type of cryptography.
They are usually implemented with the group of points of an elliptic curve over a finite field, while for the pairing the most common choices are the Tate and Weil pairings, in their modified version so that $e(g_1,g_2) \neq 1$.
For a detailed analysis of these groups, the curves to use in an implementation, and bilinear pairings in general see~\cite{lynn2007implementation,costello2012pairings}.
The most used and studied types of bilinear groups are the ones with prime order $p$.

\begin{definition}[Pairing]
  Let $\G, \Gg, \Gt$ be groups of the same prime order $p$.
  A pairing is a bilinear map $e$ such that $e: \G \times \Gg \rightarrow \Gt$ has the following properties:
  \begin{itemize}
    \item Bilinearity: $\forall g \in \G, h \in \Gg, \forall a, b \in \mathbb{Z}_p, \quad e(g^a, h^b) = e(g, h)^{ab}$.
    \item Non-degeneracy: for $g_1, g_2$ generators of $\G, \Gg$ respectively, $e(g_1, g_2)\neq 1_{\Gt}$.
  \end{itemize}
\end{definition}

\begin{definition}[Bilinear Groups]
  $\G, \Gg, \Gt$ are \emph{Bilinear Groups} if the conditions above hold and the group operations in $\G, \Gg$, and $\Gt$ as well as the bilinear map $e$ are efficiently computable.
  $\Gt$ is also called the \emph{target group}.
\end{definition}
In the remainder of this section $\G, \Gg$ and $\Gt$ are understood.

\subsubsection{Mapping of the target group}
In the protocol we use elements of the target group $\Gt$ as the source from which we derive the bitstrings that are xored with the plaintext to encrypt it.
The goal is to have a \emph{one-time-pad}-like encryption, therefore we need these bitstrings of length $\delta$ to be uniformly distributed in $\Fb^\delta$ when they are derived from elements of $\Gt$ chosen uniformly at random.

\begin{definition}[Uniform Mapping]
  Let $\Gt$ be a target group of order $p$.
  A map $\phi: \Gt \rightarrow \Fb^\delta$ is a \emph{uniform mapping} of $\Gt$ of size $\delta$ if it is efficiently computable and there is no probabilistic polynomial-time algorithm $\B: \Fb^\delta \rightarrow \{0, 1\}$ has more than negligible advantage:
  \begin{equation}
    Adv_{\B} = \Big|\text{Pr}\left[\B(\phi(g)) = 1\right] - \text{Pr}\left[\B(s)= 1\right]\Big|
  \end{equation}
  when $g$ is chosen uniformly at random in $\Gt$ and $s$ is chosen uniformly at random in $\Fb^\delta$.
\end{definition}

Depending on the group and its size, there could be encoding of its elements that are actually uniformly distributed in the bitstring space, however in practical implementations a cryptographic hash function could be used as a suitable uniform mapping.

\subsection{Decisional Bilinear Diffie-Hellman Assumption}
The Decisional Bilinear Diffie-Hellman (BDH) assumption is the basilar assumption used for proofs of indistinguishability in pairing-based cryptography.
It has been first introduced in \cite{boneh2001identity} by Boneh and Franklin in the case of symmetric pairings (where $\G = \Gg$) and then widely used in a variety of proofs, including the one of the first concrete ABE scheme in \cite{goyal2006attribute}.
Here we use a definition specific for asymmetric pairings, introduced in~\cite{boyen2008uber}.
We use an asymmetric pairing because in practical implementations they guarantee the highest level of security for a given computational cost (see e.g.~\cite{barreto2004efficient,uzunkol2018still,chatterjee2011cryptographic})
\\

Let $\alpha, \beta, \gamma, z \in \mathbb{Z}_p$ be chosen at random and $g_1, g_2$ be generators of the bilinear groups $\G, \Gg$ respectively.
The decisional Bilinear Diffie-Hellman (BDH) problem consists in constructing an algorithm
$$
  \B(A = g_1^\alpha, B=g_2^\beta, C=g_2^\gamma, T) \rightarrow \{0,1\}
$$
to efficiently distinguish between the tuples $(A,B,C,e(g_1,g_2)^{\alpha \beta \gamma})$ and \linebreak[4] $(A,B,C,e(g_1,g_2)^{z})$, outputting respectively 1 and 0.
The advantage of $\mathcal{B}$ in this case is clearly written as:
$$
Adv_{\B} = \Big|\text{Pr}\left[\B(A, B, C, \egg^{\alpha \beta \gamma}) = 1\right] - \text{Pr}\left[\B(A, B, C, \egg^z )= 1\right]\Big|
$$

\noindent where the probability is taken over the random choice of the generators $g_1, g_2$, of $\alpha, \beta, \gamma, z$ in $\mathbb{Z}_p$, and the random bits possibly consumed by $\B$ to compute the response.

\begin{definition}[BDH Assumption] \label{BDH}
  The decisional BDH assumption holds if no probabilistic polynomial-time algorithm $\B$ has a non-negligible advantage in solving the decisional BDH problem.
\end{definition}

\subsection{Interactive Diffie-Hellman}
Interactive assumptions are usually stronger than their static counterparts, since the solver has more control over the parameters.
This unfortunately means that they give weaker security, but sometimes it is nearly impossible to reduce highly interactive protocols to static assumptions, so interactive assumptions are needed.

\subsubsection{Interactive Decisional Diffie-Hellman Assumption}
Let $\C$ be a challenger that chooses $\alpha, \beta, z \in \mathbb{Z}_p$ at random and $g$ be a generator of a group $\mathbb{G}$ of prime order $p$.
The \emph{Interactive Decisional Diffie-Hellman} (IBDDH) problem consists in constructing an algorithm $\B(\C) \rightarrow \{0,1\}$ that interacts with the challenger in the following way:
\begin{itemize}
  \item $\C$ gives to $\B$ the values $A=g^\alpha, B=g^\beta$;
  \item $\B$ chooses an exponent $0 \neq s \in \Zp$ and sends to the challenger the value $S=B^{\frac{1}{s}}$;
  \item $\C$ flips a random coin $r \in \{0, 1\}$ and answers with $Z=S^\alpha = g^{\frac{\alpha \beta}{s}}$ if $r=0$, $Z=g^z$ if $r=1$;
  \item $\B$, given $A, B, S, Z$, outputs a guess $r'$ of $r$.
\end{itemize}

The advantage of $\mathcal{B}$ in this case is clearly written as:
\[
Adv_{\B} = \Big|\text{Pr}\left[\B(A, B, S, g^{\frac{\alpha \beta}{s}}) = 0\right] - \text{Pr}\left[\B(A, B, S, g^z )= 0\right]\Big|
\]
where the probability is taken over the random choice of the generator $g$, of $\alpha, \beta, s, z$ in $\mathbb{Z}_p$, of $r\in \{0,1\}$, and the random bits possibly consumed by $\B$ to compute the response.

\begin{definition}[IDDH Assumption] \label{IDDH}
  The Interactive Decisional DH assumption holds if no probabilistic polynomial-time algorithm $\B$ has a non-negligible advantage in solving the decisional IDDH problem.
\end{definition}

Note that an adversary that can solve the IDDH problem can solve the DH problem simulating a IDDH problem and choosing $s=1$, but the converse is not true since it is not possible to adapt the DH challenge without knowing $s$.


\section{Protocol Specifications}
Following the intuition given in \Cref{intuition} we present now the proposed construction.
We formally define the various actors and interactions that constitute the protocol, and the data structure that we rely on.
In the subsequent section the security will be formally and rigorously assessed, modelling some attack scenarios and proving their infeasibility.

\subsection{Masking Shards Protocol}\label{shards}
In this protocol multiple users $U_l$ publish encrypted data on a public ledger maintained by a file keeper $F$.
To gain access to the encrypted data any service provider $P$ has to ask directly to the user for a decryption key to use in combination with some masking shards that are published on the ledger.
The file keeper periodically updates the masking shards, so that older decryption keys become useless.

The ledger also has a constant section where encrypted data is actually stored and the integrity is guaranteed via chains of hash digests.

The construction of the scheme uses finite groups $\Zp$ whose size $p$ depends on a security parameter.
Throughout the following definition it is assumed that every  \emph{exponent} $x \in \Zp$ randomly generated is neither $1$ nor $0$.

\begin{definition}[Updating Masking Shards Protocol]
  An \emph{Updating Masking Shards Protocol} for a file keeper $F$, a set of users $\{U_l\}_{1 \leq l \leq N}$ and a service provider $P$ proceeds along the following steps:
  \begin{itemize}
    \normalfont
    \item $F$ sets up the public ledger: bilinear groups $\G, \Gg$ of prime order $p$ are chosen according to a security parameter $\kappa$, along with generators $g_1 \in \G, g_2 \in \Gg$.
    Let $e$ be the pairing and $\Gt$ be the target group of the same order $p$, with uniform mapping $\phi$ of size $\delta$.
    $F$ chooses uniformly at random exponents $u_i \in \Zp$ for $1 \leq i \leq I$ where $I$ is the maximum number of shards in a data block, determined from the desired data block length $|B|$ by the formula $I = |B| / \delta$.
    Finally $F$ chooses a random time-key $s_{t_0} \in \Zp$ and publishes the initial masking shards:
    \begin{equation}
      \varepsilon_{i, t_0} = g_1^{u_i s_{t_0}} \qquad 1 \leq i \leq I.
    \end{equation}
    $F$ securely saves the value $s_{t_0}$ but can forget the exponents $u_i$.
    \item $F$ periodically updates the shards choosing a new time-key $s_{t_{j+1}} \in \Zp$ and computing
    \begin{align}
      \varepsilon_{i, t_{j+1}} = (\varepsilon_{i, t_{j}})^{\frac{s_{t_{j+1}}}{s_{t_{j}}}} \qquad 1 \leq i \leq I.
    \end{align}
    \item Each user $U_l$ chooses two private exponents $\mu_l, v_l \in \Zp$, the first is immediately used to build and publish its public key:
    \begin{equation}
      q_l = g_2^{\mu_l}.
    \end{equation}
    \item To publish an encrypted file on the ledger at a time $t_j$, a user $U_l$ requests an encryption token.
    $F$ takes the public key $q_l$ of the user and computes:
    \begin{align}
      k_{l, 0, t_j} &= q_l^{\frac{1}{s_{t_j}}} \\
      &= g_2^{\frac{\mu_l}{s_{t_j}}}. \nonumber
    \end{align}
    \item Let $b - 1$ be the index of the last block in the ledger, thus the file will be published in the $b$-th block.
    Let $m_b$ be the message (file) that $U_l$ wants to encrypt, and $I_b \delta$ its length.
    Then $U_l$ divides the message in pieces $ m_{b, i} $ of equal length $\delta$, and chooses a random exponent $k_b \in \Zp$ to compute the encrypted shards as:
    \begin{align}
      c_{b, i} &= m_{b, i} \oplus \phi\left(e(\varepsilon_{i, t_{j}}, (k_{l, 0, t_j})^{k_b})\right) \qquad 1 \leq i \leq I_b \\
      &= m_{b, i} \oplus \phi\left( e(g_1^{u_i s_{t_j}}, g_2^{\frac{k_b \mu_l}{s_{t_j}}})\right) \nonumber \\
      &= m_{b, i} \oplus  \phi\left(\egg^{u_i k_b \mu_l}\right). \nonumber
    \end{align}
    If the message length is not a multiple of $\delta$ then it is sufficient to truncate the last pad $\phi\left(\egg^{u_{I_b} k_b \mu_l}\right)$ to match the length of the last piece $m_{b, I_b}$, so that the final ciphertext $c_b = c_{b,1}||\ldots||c_{b, I_b}$ has the lame length as the plaintext $m_b$.
    \item In addition to the encrypted shards, $U_l$ computes the encapsulated key using its second secret exponent:
    \begin{align}
    k_{b, 1, t_j} &= (k_{l, 0, t_j})^{\frac{v_l k_b}{\mu_l}} \\
    &= \left(g_2^{\frac{\mu_l}{s_{t_j}}}\right)^{\frac{v_l k_b}{\mu_l}} \nonumber \\
    &= g_2^{\frac{v_l k_b}{s_{t_j}}}, \nonumber
    \end{align}
    $U_l$ can forget the exponent $k_b$ once this key has been computed.
    \item Let $\bar{i} \equiv b \mod I$.
    Then $F$ computes the control shard:
    \begin{align}
    c_{b} &= \phi\left(e(\varepsilon_{\bar{i}, t_{j}}, k_{b, 1, t_j})\right) \\
    &= \phi\left(e(g_1^{u_{\bar{i}} s_{t_j}}, g_2^{\frac{k_b v_l}{s_{t_j}}})\right) \nonumber \\
    &= \phi\left(\egg^{u_{\bar{i}} k_b v_l}\right). \nonumber
    \end{align}
    \item $U_l$ sends to $F$ the digest $h(m_b)$ of the message through a secure hash function $h$, the ciphertext $c_b$, the control shard $c^*_b$ and the encapsulated key $k_{b, 1, t_j}$.
    $F$ inserts into the next data block of the \emph{static public chain} $c_b, h(m_b), c^*_b$; and inserts in the \emph{updating ledger} the encapsulated key $k_{b, 1, t_j}$.
    At this point the data block is completed with a \emph{hash-link} to the previous block and an integrity warranty that involves the hash digest of its contents (see~\Cref{block structure}).
    \item Once that at least one file has been published, $F$ has to periodically update not only the masking shards but also the encapsulated keys.
    The key update is similar to the shard update:
    \begin{align}
      k_{b, 1, t_{j+1}} &= (k_{b, 1, t_{j}})^{\frac{s_{t_{j}}}{s_{t_{j+1}}}} \\
      &= \left(g_2^{\frac{v_l k_b}{s_{t_{j}}}}\right)^{\frac{s_{t_{j}}}{s_{t_{j+1}}}} \nonumber \\
      &= g_2^{\frac{v_l k_b}{s_{t_{j+1}}}}. \nonumber
    \end{align}
    Where $s_{t_{j+1}} \in \Zp$ is the same time-key used to update the shards.
    Note that $s_{t_j}$ could and should be forgotten once the update has been completed.
    \item Let $P$ be a service provider that needs access to the file $m_b$, and therefore asks for permission to the owner of the file $U_l$.
    To grant a one-time permission (or better permission until the next update) $U_l$ computes an unlocked key valid for the current time $t_j$.
    $U_l$ retrieves from the updating ledger the encapsulated key $k_{b, 1, t_j}$ and computes:
    \begin{align}
      k_{b, 2, t_j} &= (k_{b, 1, t_j})^{\frac{\mu_l}{v_l}} \\
      &= \left(g_2^{\frac{v_l k_b}{s_{t_j}}}\right)^{\frac{\mu_l}{v_l}} \nonumber \\
      &= g_2^{\frac{\mu_l k_b}{s_{t_j}}}. \nonumber
    \end{align}
    \item With the unlocked key $P$ can decrypt the encrypted shards computing:
    \begin{align}
      m_{b, i}' &= c_{b, i} \oplus \phi\left(e(\varepsilon_{i, t_{j}}, k_{b, 2, t_j})\right) \qquad 1 \leq i \leq I_b \\
      &= m_{b, i} \oplus \phi\left(\egg^{u_i k_b \mu_l}\right) \oplus \phi\left(e(g_1^{u_i s_{t_j}}, g_2^{\frac{\mu_l k_b}{s_{t_j}}})\right) \nonumber \\
      &= m_{b, i} \oplus \phi\left(\egg^{u_i k_b \mu_l}\right) \oplus \phi\left(\egg^{u_i \mu_l k_b}\right) \nonumber \\
      &= m_{b, i}. \nonumber
    \end{align}
    Afterwards $P$ can check the integrity of the decryption comparing the hash digest of the decrypted message $h(m'_b)$ with the digest included in the $b$-th block of the static chain.
\end{itemize}
\end{definition}

\subsection{Data Block structure}\label{block structure}
Encrypted data created by the users is collected and organized into data blocks.
The File Keeper maintains and publishes a variable portion of the ledger (it could be conceptually separated from the \emph{static ledger}, as discussed in~\Cref{intuition}) that contains at any time $t$:
\begin{itemize}
  \item the masked shards and their index:
  \begin{equation}
  (\varepsilon_{i, t}, i)_{1 \leq i \leq I};
  \end{equation}
  \item the encapsulated keys and the index of the data block where the corresponding encrypted pieces are stored:
  \begin{equation}
  (k_{b, 1, t}, b)_{b \geq 1}.
  \end{equation}
\end{itemize}
All these elements are kept constantly updated.

The other part of the ledger is more akin to a classical blockchain and is comprised of constant data blocks linked together.
With ``constant'' we mean that new data blocks may be added any time, but the old data blocks are never changed.
Each data block $B_b$ contains:
\begin{itemize}
  \item the ciphertext, the digests of the original cleartext, and the control shard:
  \begin{equation}
    D_b = (c_{b},h(m_{b}),c^*_{b}); \label{complete-block}
  \end{equation}
  \item the hash of the previous data block $h(B_{b-1})$;
\end{itemize}
Note that by design the index of the masking shard associated to the control shard covers the whole range $1 \leq i \leq I$.
These pieces are needed to check the integrity of the data stored in the updating part of the ledger (encapsulated keys and masking shards).
In fact let $\bar{i} = b \mod I$, then for every time $t$ it should hold:
\begin{align}
c^*_{b} &= \phi\left(e(\varepsilon_{\bar{i}, t}, k_{b, 1, t})\right) \\
&= \phi\left(e(g_1^{u_{\bar{i}} s_{t}}, g_2^{\frac{k_b v_l}{s_{t}}})\right) \nonumber \\
&= \phi\left(\egg^{u_{\bar{i}} k_b v_l}\right). \nonumber
\end{align}
Therefore any observer could check the coherence of the updating ledger (and consequently the behaviour of the file keeper $F$).

In alternative to the content defined in~\Cref{complete-block}, a more efficient blockchain may be built excluding the actual encrypted data from the data blocks, retaining only its digest.
That is the bulk of data is stored in distributed databases, while their hash is kept on the ledger to guarantee the integrity.
This approach reduces consistently the size of the data blocks, that therefore can be more widely distributed.
This shrunk data block would therefore substitute $D_b$ as defined in~\Cref{complete-block} with:
\begin{equation}
  D'_b = (h(c_{b}),h(m_{b}),c^*_{b}).
\end{equation}

The last thing required is a cryptographic warranty of immutability of the static data blocks.
If this protocol is combined with a pre-existing blockchain then the latter may be exploited to guarantee data integrity: it is sufficient to compute the digest:
\begin{equation} \label{block-hash-pled}
d_b = h(h(B_{b-1}), D_b),
\end{equation}
or $ d'_{b} = h(h(B_{b-1}), D'_b) $ with the shrunken variant, and simply embedding it in a transaction registered on the aforementioned blockchain, effectively time-stamping it (e.g. using  \texttt{OPRETURN} in Bitcoin).

Alternatively, in absence of such an anchor, the warranty has to be included in the static data blocks.
Here we propose three different approaches, that are to include either:
\begin{itemize}
  \item [$\triangleright$] the signature of the user (owner of the data) on $d_b$;
  \item [$\triangleright$] a proof of work involving $d_b$;
  \item [$\triangleright$] a signature made by a third party (or a group or multi-party signature) on $d_b$;
\end{itemize}
These solutions all have pros and cons to their adoption, the optimal choice is probably a combination of the three.

The user signature would give proof of ownership of the data, on the other hand without the signature the content remains fully anonymous to anyone besides $F$, that can act as a proxy and forward the access request of $P$ to the relevant user $U_l$ without disclosing its identity.
To guarantee the end-to-end property of the protocol $P$ should have a public key known to (or retrievable by) every user.
So when $F$ forwards a request by $P$ to $U_l$, the user can encrypt the unlocked key with this public key, so that only $P$ can use it.

The proof of work would give increasingly stronger proof of integrity to older data, in the sense that the more blocks are added to the ledger, the more infeasible it gets to manipulate older data maintaining the consistency of the chain.

The third party signature gives immediate proof of integrity of the block, in alternative or addition to the proof of work.
The trade-off between these solutions is that a signature gives instantaneous integrity evidence but relies on the honesty of the signer (possibly mitigated with multi-party signatures), whereas a proof of work is more expensive and provides sufficiently reliable security only for older blocks, but it does not rely on anyone's honesty.
\newline


\section{Security Model}\label{pled-sec}
The goals of the protocol is to achieve the following security properties:
\begin{description}
  \item [End-to-end encryption] The file keeper must not be able to read the plaintext message at any time.
  \item [One-time access] A service provider should be able to read a plaintext message at the time $t$ if and only if the file owner authorizes them with an unlocked key for the time $t$.
\end{description}


\subsection{Security against Outsiders and Service Providers}
  The security of the protocol is proven in terms of chosen-plaintext indistinguishability, the security game is formally defined as follows.\\

\begin{definition}[Security Game] \label{pledsgo}
  The security game for an updating masking protocol proceeds as follows:
  \begin{description}
    \normalfont
    \item [Init] The adversary $\A$ chooses a number of users $N$ that will encrypt files and the maximum number of masking shards $I$.
    \item [Setup] For each user $U_l$, with $1\leq l \leq N$ the challenger $\C$ sets up a public key $q_l$, and takes the role of the file keeper by publishing the initial masking shards $\varepsilon_{i, t_0}$, for $1 \leq i \leq I$.
    \item [Phase 0] The adversary may request updates of the masking shards $\varepsilon_{i, t_j}$, for ${1 \leq i \leq I}$, $0 \leq j < n$.
    \item [Commit] The users commit to a key before creating a ciphertext by publishing the encryption tokens $k_{l, 0, t_n}$ and the encapsulated keys $k_{b, 1, t_n}$.
    \item [Phase 1] The adversary for each time $n \leq j < n^*$ may request updates of the encapsulated keys $k_{b, 1, t_j}$, and either  the corresponding unlocked keys $k_{b, 2, t_j}$, or the masking shards $\varepsilon_{i, t_j}$, but not both.
    \item [Challenge] Let $\delta$ be the size of the uniform mapping $\phi$ of the target group $\Gt$.
    For each $1\leq l \leq N$ the adversary chooses two messages $m_{l,0}, m_{l, 1}$ of length $I_l \delta$ and sends them to the challenger, which flips a random coin $r_l \in \{0, 1\}$ and publishes the encryption $(c_{b, i}, 1 \leq i \leq I_l)$ of the message $m_{l, r_l}$.
    \item [Phase 2] Phase 1 is repeated for $n* \leq j < n'$.
    \item [Guess] The adversary chooses an $\bar{l}$ such that $1\leq \bar{l} \leq N$ and outputs a guess $r_{\bar{l}}'$ of $r_{\bar{l}}$.
  \end{description}
\end{definition}

\begin{definition}[Updating Masking Security]
  An Updating Masking protocol with security parameter $\xi$ is CPA secure if for all probabilistic polynomial-time adversaries $\A$, there exists a negligible function $\psi$ such that:
  \begin{equation}
    Pr[r_{\bar{l}}' = r_{\bar{l}}] \leq \frac{1}{2} + \psi(\xi)
  \end{equation}
\end{definition}

The scheme is proved secure under the BDH (Decisional Bilinear Diffie-Hellman) assumption (Definition \ref{BDH})
in the security game defined above.

The security  is provided by the following theorem.

\begin{theorem} \label{pledo}
  If an adversary can break the scheme, then a simulator can be constructed to play the decisional BDH game with non-negligible advantage.
\end{theorem}
\begin{proof}
  Suppose there exists a polynomial-time adversary $\A$, that can attack the scheme in the Selective-Set model with advantage $\epsilon$.
  Then a simulator $\B$ can be built that can play the Decisional BDH game with advantage $\epsilon/2$.
  The simulation proceeds as follows.

  \paragraph*{Init}
    The adversary chooses the number of users $N$ and the maximum number of masking shards $I$, the simulator takes in a BDH challenge $g_1, g_2$, $A=g_1^\alpha$, $B=g_2^\beta, C=g_2^\gamma, T$.

  \paragraph*{Setup}
  The simulator chooses random $\mu_l', \omega_l, \mu_l \in \Zp$ for $1 \leq l \leq N$, $u_i' \in \Zp$ for $1 \leq i \leq I$.
  Then it implicitly sets:
  \begin{equation}
    \mu_l := \mu_l' \beta \qquad v_l:=\frac{\omega_l \beta}{\gamma} \qquad 1 \leq l \leq N, \qquad u_i:= u_i'\alpha \qquad 1 \leq i \leq I.
  \end{equation}
  So it publishes the public keys of the users and the initial masking shards:
  \begin{align}
    q_l&= B^{\mu_l'} \qquad 1 \leq l \leq N, \qquad &\varepsilon_{i, t_0} &= A^{u_i' s_{t_0}} \qquad 1 \leq i \leq I\\
    &=g_2^{\beta \mu_l'} & &= g_1^{\alpha u_i' s_{t_0}} \nonumber\\
    &=g_2^{\mu_l} & &= g_1^{u_i s_{t_0}}. \nonumber
  \end{align}

  \paragraph*{Phase 0}
  In this phase the simulator answers to update queries of the masking shards.
  For each time $0 \leq j \leq n$ it chooses uniformly at random $s_{t_j} \in \Zp$ and computes:
  \begin{align}
    \varepsilon_{i, t_j} &= A^{u_i' s_{t_j}} \qquad 1 \leq i \leq I \\
    &=g_1^{\alpha u_i' s_{t_j}} \nonumber \\
    &=g_1^{u_i s_{t_j}}. \nonumber
  \end{align}

  \paragraph*{Commit}
  In this phase users commit to a key before creating a ciphertext, by publishing their encryption tokens and encapsulated keys.
  Note that each user may commit at a different time, but for simplicity we suppose that they commit all together.
  Moreover from now on the indexes $b$ and $l$ will be identified, since for the purposes of this proof only one encryption per user is considered.
  
  To simulate the commitment $\B$ chooses uniformly at random $k_l' \in \Zp$ and implicitly sets $k_l:= k_l'\gamma$ for $1 \leq l \leq N$.
  Furthermore it chooses $s_{t_n} \in \Zp$, then it can compute:
  \begin{align}
    k_{l, 0, t_n} &= B^{\frac{\mu_l'}{s_{t_n}}},  &k_{b, 1, t_n} &= B^{\frac{\omega_l k_l'}{s_{t_n}}} \qquad 1 \leq l \leq N \\
    &=g_2^{\frac{\beta \mu_l'}{s_{t_n}}} & &=g_2^{\frac{\beta \omega_l \gamma k_l'}{\gamma s_{t_n}}} \nonumber \\
    &=g_2^{\frac{\mu_l}{s_{t_n}}} & &=g_2^{\frac{v_l k_l}{s_{t_n}}}. \nonumber
  \end{align}

  \paragraph*{Phase 1}
  In this phase the adversary for each time $n \leq j < n^*$ may request updates of the encapsulated keys $k_{b, 1, t_j}$, and either  the corresponding unlocked keys $k_{b, 2, t_j}$, or the masking shards $\varepsilon_{i, t_j}$, but not both.
  If the adversary asks for the unlocked keys the simulator chooses at random $s_{t_j}' \in \Zp$ and implicitly sets $s_{t_j} := s_{t_j}' \beta$.
  So it can compute:
  \begin{align}
    k_{b, 1, t_j} &= g_2^{\frac{\omega_l k_l'}{s_{t_j}'}}, & k_{b, 2, t_j} &= C^{\frac{k_l'\mu_l'}{s_{t_j}'}} \qquad 1 \leq l \leq N \\
    &= g_2^{\frac{\beta \omega_l \gamma k_l'}{\gamma \beta s_{t_j}'}} & &= g_2^{\frac{\gamma k_l' \beta \mu_l'}{ \beta s_{t_j}'}} \nonumber \\
    &= g_2^{\frac{v_l k_l}{s_{t_j}}} & &= g_2^{\frac{k_l \mu_l}{s_{t_j}}}. \nonumber
  \end{align}
  Otherwise, if the adversary asks for the masking shards, it chooses $s_{t_j} \in \Zp$ and computes
  \begin{align}
    k_{b, 1, t_j} &= B^{\frac{\omega_l k_l'}{s_{t_j}}} \qquad 1 \leq l \leq N, & \varepsilon_{i, t_j} &= A^{u_i' s_{t_j}} \qquad 1 \leq i \leq I \\
    &=g_2^{\frac{\beta \omega_l \gamma k_l'}{\gamma s_{t_j}}} & &=g_1^{\alpha u_i' s_{t_j}} \nonumber \\
    &=g_2^{\frac{v_l k_l}{s_{t_j}}} & &=g_1^{u_i s_{t_j}} \nonumber
  \end{align}

  \paragraph*{Challenge}
  For each $1\leq l \leq N$ the adversary sends two messages $m_{l,0}, m_{l, 1}$of length $I_l \delta$.
  The simulator flips $N$ random coins $r_l \in \{0, 1\}$ then creates the ciphertexts as:
  \begin{align}
    c_{b, i} &:= m_{l, r_l, i} \oplus \phi\left(T^{u_i'k_l'\mu_l'}\right)\\
    &\stackrel{*}{=} m_{l, r_l, i} \oplus \phi\left(\egg^{\gamma k_l' \beta \mu_l' \alpha u_i'}\right) \nonumber\\
    &= m_{l, r_l, i} \oplus \phi\left(\egg^{k_l \mu_l u_i}\right)  & 1 \leq i \leq I_l, \quad 1 \leq l \leq N \nonumber
  \end{align}
  where the equality $\stackrel{*}{=}$ holds if and only if the BDH challenge was a valid tuple (i.e. $T$ is non-random).

  \paragraph*{Phase 2}~
    During this phase the simulator acts exactly as in \emph{Phase~1}.

  \paragraph*{Guess}
      The adversary will eventually select a user $\bar{l}$ and output a guess $r_{\bar{l}}'$ of $r_{\bar{l}}$.
      The simulator then outputs $0$ to guess that $T = \egg^{\alpha \beta \gamma}$ if $r_{\bar{l}}' = r_{\bar{l}}$; otherwise, it outputs $1$ to indicate that it believes $T$ is a random group element in $\Gt$.
      In fact when $T$ is not random the simulator $\mathcal{B}$ gives a perfect simulation so it holds:
      $$
        Pr\left[\mathcal{B}\left(\vec{y},T=\egg^{\alpha \beta \gamma}\right)=0\right] = \frac{1}{2} + \epsilon
      $$
      On the contrary when $T$ is a random element $R \in \Gt$, since $\phi$ is a uniform mapping then the messages $m_{r_l}$ are completely hidden from the adversary point of view, so:
      $$
        Pr\left[\mathcal{B}\left(\vec{y},T=R\right)=0\right] = \frac{1}{2}
      $$
      Therefore, $\mathcal{B}$ can play the decisional BDH game with non-negligible advantage$~\frac{\epsilon}{2}$.
\end{proof}

\subsection{Security Against the File Keeper}
  In this section we describe the \emph{end to end} privacy of the protocol testing its robustness in scenarios where the File Keeper itself tries to read the content of the encrypted data stored on the ledger.
  The security of the protocol will again be proven in terms of chosen-plaintext indistinguishability, but there is a distinction between two scenarios.
  
  In the first one the security will be proven using the standard Decisional Bilinear Diffie Hellman Assumption, but we assume that the File Keeper is not malicious and that it takes over the role after shards initialization.
  That is the protocol is initialized and then the relevant information is passed to the File Keeper that subsequently fulfils its role following the protocol (but trying to decrypt the files).

  In the second scenario the File Keeper independently sets up the masking shards and freely interacts with the users, but  in this case an interactive assumption (Interactive Decisional Diffie-Hellman or IDDH) defined in~\ref{IDDH}
   is necessary to prove the security.

  The security game of the first scenario is formally defined as follows.\\

\begin{definition}[Curious File Keeper Security Game]
  The security game for the protocol with a curious File Keeper proceeds as follows:
  \begin{description}
    \normalfont
    \item [Init] The adversary $\A$ chooses a number of users $N$ that will encrypt files and the maximum number of masking shards $I$.
    \item [Setup] For each user $U_l$, with $1\leq l \leq N$ the challenger $\C$ sets up a public key $q_l$, and initializes the masking shards $\varepsilon_{i, t_0}$, for $1 \leq i \leq I$, giving also $s_{t_0}$ to $\A$.
    \item [Commit] $\A$ asks the users commit to a key before creating a ciphertext giving them encryption tokens $k_{l, 0, t_j}$, $\C$ responds publishing encapsulated keys $k_{b, 1, t_j}$.
    \item [Challenge] Let $\delta$ be the size of the uniform mapping $\phi$ of $\Gt$.
    For each ${1\leq l \leq N}$ the adversary chooses two messages $m_{l,0}, m_{l, 1}$ of length $I_l \delta$ and sends them to the challenger that flips a random coin $r_l \in \{0, 1\}$ and publishes the encryption $(c_{b, i}, 1 \leq i \leq I_l)$ of the message $m_{l, r_l}$.
    \item [Guess] The adversary chooses an $\bar{l}$ such that $1\leq \bar{l} \leq N$ and outputs a guess $r_{\bar{l}}'$ of $r_{\bar{l}}$.
  \end{description}
\end{definition}

\begin{definition}[Security with a Curious File Keeper]
  An Updating Masking protocol with security parameter $\xi$ is CPA secure with a Curious File Keeper if for all probabilistic polynomial-time adversaries $\A$, there exists a negligible function $\psi$ such that:
  \begin{equation}
    Pr[r_{\bar{l}}' = r_{\bar{l}}] \leq \frac{1}{2} + \psi(\xi)
  \end{equation}
\end{definition}

The security is provided by the following theorem.

\begin{theorem}\label{pledfk1}
  If an adversary can break the scheme, then a simulator can be constructed to play the decisional BDH game with non-negligible advantage.
\end{theorem}
\begin{proof}
  Suppose there exists a polynomial-time adversary $\A$, that can attack the scheme in the Selective-Set model with advantage $\epsilon$.
  Then a simulator $\B$ can be built that can play the Decisional BDH game with advantage $\epsilon/2$.
  The simulation proceeds as follows.

  \paragraph*{Init}
    The adversary chooses the number of users $N$ and the maximum number of masking shards $I$, the simulator takes in a DBDH challenge $g_1, g_2$, ${A=g_1^\alpha}, {B=g_2^\beta}$, ${C=g_2^\gamma}, T$.

  \paragraph*{Setup}
  The simulator chooses random $\mu_l', \omega_l, \mu_l \in \Zp$ for $1 \leq l \leq N$, $u_i' \in \Zp$ for $1 \leq i \leq I$, $s_{t_0} \in \Zp$.
  Then it implicitly sets:
  \begin{equation}
    \mu_l := \mu_l' \beta \qquad v_l:=\frac{\omega_l \beta}{\gamma} \qquad 1 \leq l \leq N, \qquad u_i:= u_i'\alpha \qquad 1 \leq i \leq I.
  \end{equation}
  So it publishes the public keys of the users and the initial masking shards:
  \begin{align}
    q_l&= B^{\mu_l'} \qquad 1 \leq l \leq N, \qquad &\varepsilon_{i, t_0} &= A^{u_i' s_{t_0}} \qquad 1 \leq i \leq I\\
    &=g_2^{\beta \mu_l'} & &= g_1^{\alpha u_i' s_{t_0}} \nonumber\\
    &=g_2^{\mu_l} & &= g_1^{u_i s_{t_0}} .\nonumber
  \end{align}
  Moreover the exponent $s_{t_0}$ is given to the adversary.

  \paragraph*{Commit}
  In this phase the adversary asks the users to commit to a key before creating a ciphertext, by giving them encryption tokens and obtaining encapsulated keys in return.
  For every user $l$ the adversary may choose a different time $t_{j_l}$ in which the commitment takes place.
  To formulate the query $\A$ chooses a random exponent $s_{t_{j_l}}$ and computes the encryption token:
  \begin{align}
  k_{l, 0, t_{j_{l}}} &= q_l^{\frac{1}{s_{t_{j_{l}}}}}\\
    &=B^{\frac{\mu_l'}{s_{t_{j_{l}}}}} \nonumber \\
    &=g_2^{\frac{\beta \mu_l'}{s_{t_{j_{l}}}}} \nonumber \\
    &=g_2^{\frac{\mu_l}{s_{t_{j_{l}}}}}.\nonumber
  \end{align}
  To simulate the answer $\B$ chooses uniformly at random $k_l' \in \Zp$ and implicitly sets $k_l:= k_l'\gamma$ for $1 \leq l \leq N$.
  Then it can compute:
  \begin{align}
    k_{b, 1, t_{j_{l}}} &= (k_{l, 0, t_{j_{l}}})^{\frac{\omega_l k_l'}{\mu_l'}} \qquad 1 \leq l \leq N \\
    &=g_2^{\frac{\beta \mu_l'}{s_{t_{j_{l}}}} \frac{\omega_l k_l' \gamma}{\mu_l' \gamma}} \nonumber \\
    &=g_2^{\frac{\beta}{\gamma}\omega_l \frac{\gamma k_l'}{s_{t_{j_{l}}}}} \nonumber \\
    &=g_2^{\frac{v_l k_l}{s_{t_{j_{l}}}}}. \nonumber
  \end{align}

  \paragraph*{Challenge}
  For each $1\leq l \leq N$ the adversary sends two messages $m_{l,0}, m_{l, 1}$of length $I_l \delta$.
  The simulator flips $N$ random coins $r_l \in \{0, 1\}$ then creates the ciphertexts as:
  \begin{align}
    c_{b, i} &:= m_{l, r_l, i} \oplus \phi\left(T^{u_i'k_l'\mu_l'}\right)\\
    &\stackrel{*}{=} m_{l, r_l, i} \oplus \phi\left(\egg^{\gamma k_l' \beta \mu_l' \alpha u_i'}\right) \nonumber\\
    &= m_{l, r_l, i} \oplus \phi\left(\egg^{k_l \mu_l u_i}\right)  & 1 \leq i \leq I_l, \quad 1 \leq l \leq N \nonumber.
  \end{align}
  Where the equality $\stackrel{*}{=}$ holds if and only if the BDH challenge was a valid tuple (i.e. $T$ is non-random).

  \paragraph*{Guess}
      The adversary will eventually select a user $\bar{l}$ and output a guess $r_{\bar{l}}'$ of $r_{\bar{l}}$.
      The simulator then outputs $0$ to guess that $T = \egg^{\alpha \beta \gamma}$ if $r_{\bar{l}}' = r_{\bar{l}}$; otherwise, it outputs $1$ to indicate that it believes $T$ is a random group element in $\Gg$.
      In fact when $T$ is not random the simulator $\mathcal{B}$ gives a perfect simulation so it holds:
      $$
        Pr\left[\mathcal{B}\left(\vec{y},T=\egg^{\alpha \beta \gamma}\right)=0\right] = \frac{1}{2} + \epsilon
      $$
      On the contrary when $T$ is a random element $R \in \Gt$, since $\phi$ is a uniform mapping then the messages $m_{r_l}$ are completely hidden from the adversary point of view, so:
      $$
        Pr\left[\mathcal{B}\left(\vec{y},T=R\right)=0\right] = \frac{1}{2}
      $$
      Therefore, $\mathcal{B}$ can play the decisional BDH game with non-negligible advantage$~\frac{\epsilon}{2}$.
\end{proof}

During the simulation in this proof the update of the masking shards and encapsulated keys has not been explicitly considered because the adversary has the role of the File Keeper, that is the party in charge of such operations, so the simulator is not involved.

Note also that for the encryption in the challenge phase the simulator does not use neither the masking shards nor the encapsulated key, so the update of these elements is not strictly necessary (although the simulator should request them even without using them just for the sake of a thorough simulation).
This is possible only because the simulator controls the initialization of the masking shards, and this limits the possibilities for the attacker (e.g. choosing particular values for $s_t$).

To consider a more powerful adversary and take account of possible interaction the security game can be modified in this way:
\begin{itemize}
  \item in the setup phase the simulator does not initialize the masking shards, the adversary has complete control over them
  \item in the challenge phase the simulator asks for the updated version of the masking shards and the encapsulated keys before computing the encryption, moreover the adversary may ask that the encryptions take place at different times.
\end{itemize}

Now the security against this more powerful File Keeper can be given with the following theorem.

\begin{theorem}\label{pledfk2}
  If an adversary taking the role of the file keeper can break the scheme, then a simulator can be constructed to play the IDDH game in $\Gg$ with non-negligible advantage.
\end{theorem}

\begin{proof}
  The proofs is almost identical to the proof of Theorem~\ref{pledfk1}, only the following tweaks are necessary:
  \begin{itemize}
    \item the simulator starts initializing the IDDH challenge, obtaining the elements $A=g_2^\alpha, B=g_2^\beta$;
    \item during the setup phase the simulator does not have to initialize the masking shards so it does not need the element $C$ and can therefore use $A$ in the the same way it used $C$ in the previous simulations.
    \item during the challenge phase suppose that the adversary requests at time $t_j$ the encryption of the user $l$, then the simulator asks the adversary for the value of the updated masking shards $\varepsilon_{i, t_j}$ for $1 \leq i \leq I$, and the updated encapsulated key $k_{b, 1, t_{j}}$.
    Then the simulator $\B$ interacts with the challenger $\C$ of the IDDH game sending the value:
    \begin{align}
      S &= k_{b, 1, t_{j}}^{\frac{1}{k_b'\omega_l}} \\
      &\stackrel{\bullet}{=} (g^{\frac{\beta \omega_l k_b'}{s_{t_j}}})^{\frac{1}{k_b'\omega_l}} \nonumber \\
      &= g^{\frac{\beta}{s_{t_j}}} \nonumber \\
      &= B^{\frac{1}{s_{t_j}}}. \nonumber
    \end{align}
    $\C$ answers with $Z$ and $\B$ proceeds with the encryption computing
    \begin{align}
    c_{b, i} &:= m_{l, r_l, i} \oplus \phi\left(e(\varepsilon_{i, t_j}, Z^{k_l'\mu_l'})\right)\\
    &\stackrel{\bullet}{=} m_{l, r_l, i} \oplus \phi\left(e(g_1^{u_i s_{t_j}},Z^{k_l'\mu_l'})\right)\nonumber\\
    &\stackrel{*}{=} m_{l, r_l, i} \oplus \phi\left(e\left(g_1^{u_i s_{t_j}},(g_2^{\frac{\alpha \beta}{s_{t_j}}})^{k_l'\mu_l'}\right)\right) \nonumber\\
    &= m_{l, r_l, i} \oplus \phi\left(\egg^{u_i \alpha k_l' \beta \mu_l' }\right) \nonumber \\
    &= m_{l, r_l, i} \oplus \phi\left(\egg^{u_i k_l \mu_l }\right)  & 1 \leq i \leq I_l \nonumber
  \end{align}
  Where the equality $\stackrel{*}{=}$ holds if and only if in the IDDH challenge the value of the random coin tossed by $\C$ is $r=0$.
  Note that the equalities $\stackrel{\bullet}{=}$ hold if and only if the adversary followed the protocol acting as the File Keeper, however the interaction with $\C$ is valid even if this is not the case, and from the prospective of $\A$ the simulation has the same distribution of an interaction with the real protocol.
  \end{itemize}
  The rest of the simulation is identical and the same considerations hold, thus $\B$ has an advantage of $\frac{\epsilon}{2}$ playing the IDDH.
\end{proof}

\subsection{Security against other Users}
  In the previous sections the robustness of the protocol against outside attackers and even against the file keeper has been proven.
  To complete the analysis we present now some considerations about the remaining party that participates in the protocol: the users.

  Consider a setting in which the attacker interacts with the protocol as a normal user (thus requests encryption tokens, publishes encapsulated keys and ciphertexts), but then tries to distinguish the encryption of plaintexts of their choice performed by other users.
  Starting from the security game of Definition~\ref{pledsgo}, phases 1 and 2 can be modified removing unlocked key queries but adding queries for encryption tokens given a public key.
  Once the adversary has given the challenger an encapsulated key, the update of these keys may be requested, just like any other key.
  In this way we model a (possibly malicious) user, that interacts with the file keeper and can observe the chain and its evolution.
  
  This model leads to the same conclusion on the security of the system.
  Since the system is secure against the file keeper, and a user knows less than the file keeper, then the security holds also against other users.

  The proof of the security in this scenario follows directly from the proof of Theorem~\ref{pledo}, in fact without the need of simulating unlocked keys the simulator can always choose $s_t$ freely, and thus can follow slavishly the protocol when the adversary requests encryption tokens and updates of the related keys.


\section{Conclusion and comments}
The protocol presented here expands the scope of distributed ledgers, and in particular blockchain-like designs, to include the safe storage of sensitive data.

The approach used to achieve the one-time access property aims at a highly efficient revocation.
That is, the concept of masking shards is used to revoke the access to every ciphertext updating only common elements.
This means that it is not necessary to update every block, but only the shards and the encapsulated keys, that are very much shorter than the actual data, thus achieving great efficiency gains in comparison with re-encryption approaches.
Furthermore this allows to check the integrity of the ciphertexts even before decryption, and any observer can monitor the integrity of the ledger checking the coherence of the hash digests in the static chain and verifying the control shards against the data contained in the updating section.
\newline

As outlined in~\Cref{intuition} the one-time access property also relies in the impossibility for a service provider $P$ to save a key smaller than the actual data that guarantees access beyond revocation, and we noted how in practice data may be compressed, while cryptographic keys may not.
In light of these observation one may improve the efficiency of the protocol expanding the length of the pads so that fewer shards (and therefore fewer pairing computations) are required for encryption and decryption, in a sense compensating for the different compressibility of data and keys.
When customising these parameter one shall consider that the key material necessary for decryption has (almost exactly) the same length of the plaintext when the size $\delta$ of the uniform mapping is the same as the length of an optimal encoding of an element of $\G$, since the masking shards are elements of this group. 

Given the proofs of security against a curious file keeper, it follows that this role is only busy updating shards and encapsulated keys, but it is not depositary of trust in a privacy sense.
To further reduce the dependence on the file keeper, the protocol can be modified in order to employ temporary file keepers, that are only responsible to perform one update.
The passage of responsibilities can be done in multiple ways, e.g. choosing a random candidate in a given set, or voting the preferred successor.
In any case the  current $F$ chooses a random $s_t$, while the exponent $s_{t-1}$ is obtained from the previous file keeper in a safe and secure way; once the update has been completed and a successor has been nominated $F$ passes on $s_t$.
Note that this method enforces the oblivion of old exponents, preventing previous file keepers to reverse a revocation.
\newline

A final remark regards the frequency of revocation.
As presented here the protocol revokes every ciphertext at once, and while this might be convenient in terms of revocation efficiency, not every application suits this approach, in particular when it is not feasible to burden the user with frequent unlocking of encapsulated keys to restore access to revoked files.
An easy solution is to employ different sets of encryption shards, and divide the updating section of the ledger according to different frequencies of revocation.
For example a practical ledger could have a set of shards updated with medium frequency, suitable for most of the regular files, a set reserved for quick revocation of very sensitive files, and finally a set for long term accesses, that is updated only in case of necessity.
\newline

In conclusion the protocol presented here demonstrates a practical construction of a scheme that allows tightly-managed  sharing of sensitive information with one-time-access property and fast revocation, combined with full support of public auditing.
These properties makes it a perfect candidate for a cloud solution that aims to safely operate with sensitive data cooperating with multiple service providers (for example health records that have to be shared with insurance companies, medical professionals, hospitals, etc.), also considering the requirements of modern privacy legislation such as GDPR.





\bibliographystyle{splncs}
\bibliography{private-ledger}

\begin{thebibliography}{10}
\providecommand{\url}[1]{\texttt{#1}}
\providecommand{\urlprefix}{URL }

\bibitem{azaria2016medrec}
Azaria, A., Ekblaw, A., Vieira, T., Lippman, A.: Medrec: Using blockchain for
  medical data access and permission management. In: Open and Big Data (OBD),
  International Conference on. pp. 25--30. IEEE (2016)

\bibitem{barreto2004efficient}
Barreto, P.S., Lynn, B., Scott, M.: Efficient implementation of pairing-based
  cryptosystems. Journal of Cryptology  17(4),  321--334 (2004)

\bibitem{boneh2001identity}
Boneh, D., Franklin, M.: Identity-based encryption from the weil pairing. In:
  Advances in Cryptology—CRYPTO 2001. pp. 213--229. Springer (2001)

\bibitem{boyen2008uber}
Boyen, X.: The uber-assumption family. In: International Conference on
  Pairing-Based Cryptography. pp. 39--56. Springer (2008)

\bibitem{chatterjee2011cryptographic}
Chatterjee, S., Menezes, A.: On cryptographic protocols employing asymmetric
  pairings—the role of $\psi$ revisited. Discrete Applied Mathematics
  159(13),  1311--1322 (2011)

\bibitem{costello2012pairings}
Costello, C.: Pairings for beginners (2012),
  \url{http://www.craigcostello.com.au/pairings/PairingsForBeginners.pdf}

\bibitem{garay1997secure}
Garay, J.A., Gennaro, R., Jutla, C., Rabin, T.: Secure distributed storage and
  retrieval. In: Mavronicolas, M., Tsigas, P. (eds.) Distributed Algorithms.
  pp. 275--289. Springer Berlin Heidelberg, Berlin, Heidelberg (1997)

\bibitem{goyal2006attribute}
Goyal, V., Pandey, O., Sahai, A., Waters, B.: Attribute-based encryption for
  fine-grained access control of encrypted data. In: Proc. of {CCS}~06. pp.
  89--98 (2006)

\bibitem{groth2010short}
Groth, J.: Short pairing-based non-interactive zero-knowledge arguments. In:
  International Conference on the Theory and Application of Cryptology and
  Information Security. pp. 321--340. Springer (2010)

\bibitem{kosba2016hawk}
Kosba, A., Miller, A., Shi, E., Wen, Z., Papamanthou, C.: Hawk: The blockchain
  model of cryptography and privacy-preserving smart contracts. In: 2016 IEEE
  symposium on security and privacy (SP). pp. 839--858. IEEE (2016)

\bibitem{longo2018formal}
Longo, R.: Formal Proofs of Security for Privacy-Preserving Blockchains and
  other Cryptographic Protocols. Ph.D. thesis, University of Trento (2018)

\bibitem{lynn2007implementation}
Lynn, B.: On the implementation of pairing-based cryptosystems. Ph.D. thesis,
  Stanford University (2007)

\bibitem{bitcoin}
Nakamoto, S.: {Bitcoin}: a peer-to-peer electronic cash system.
  \url{https://bitcoin.org/bitcoin.pdf} (2008)

\bibitem{uzunkol2018still}
Uzunkol, O., Kiraz, M.S.: Still wrong use of pairings in cryptography. Applied
  Mathematics and Computation  333,  467--479 (2018)

\bibitem{zyskind2015decentralizing}
Zyskind, G., Nathan, O., et~al.: Decentralizing privacy: Using blockchain to
  protect personal data. In: Security and Privacy Workshops (SPW), 2015 IEEE.
  pp. 180--184. IEEE (2015)

\end{thebibliography}

\end{document}